\documentclass{snapshotmfo}
\usepackage[utf8]{inputenc}
\usepackage[dvipsnames]{xcolor}
\usepackage{mathtools}
\usepackage{enumitem}
\usepackage{amsmath,amssymb}
\usepackage{amsthm}
\usepackage{tcolorbox}
\usepackage{epigraph}
\usepackage{hyperref}
\hypersetup{
    colorlinks=true,
    linkcolor=Aquamarine,
    filecolor=Aquamarine,      
    urlcolor=Aquamarine,
    citecolor=Aquamarine,
}

\usepackage[round]{natbib}
\usepackage[USenglish]{babel}

\usepackage{ellipsis}

\DeclareMathOperator\inter{int}

\DeclareMathOperator*{\argmax}{arg\,max}

\DeclareMathOperator*{\conv}{conv}

\newtheorem{theorem}{Theorem}[section]
\newtheorem{proposition}[theorem]{Proposition}
\newtheorem{lemma}[theorem]{Lemma}
\newtheorem{remark}[theorem]{Remark}
\newtheorem{corollary}[theorem]{Corollary}
\newtheorem{claim}[theorem]{Claim}
\newtheorem{definition}[theorem]{Definition}

\author{Konstantin von Beringe \and Mark Whitmeyer\thanks{Arizona State University. We thank Ludvig Sinander for alerting us to the Dutch Book justification for Bayesianism. We are also grateful to Geoffroy de Clippel and Collin Raymond for suggestions that led to this paper and to Joseph Whitmeyer for his feedback. Emails: \href{mailto:kvonberi@asu.edu}{kvonberi@asu.edu} and \href{mailto:mark.whitmeyer@gmail.com}{mark.whitmeyer@gmail.com}.}}
\title{The Perils of Overreaction}

\begin{document}

\begin{abstract}
In order to study updating rules, we consider the problem of a malevolent principal screening an imperfectly Bayesian agent. We uncover a fundamental dichotomy between underreaction and overreaction to information. If an agent's posterior is farther away from the prior than it should be under Bayes' law, she can always be exploited by the principal to an unfettered degree: the agent's \textit{ex ante} expected loss can be made arbitrarily large. In stark contrast, an agent who underreacts (whose posterior is closer to the prior than the Bayesian posterior) cannot be exploited at all.
\end{abstract}

\newpage


\section{Exploiting Non-Bayesians}

A classical justification for Bayesianism is a Dutch Book argument: a malevolent (or at least selfish) bookie can offer a non-Bayesian a sequence of bets that guarantee the bookie a strictly positive profit (\cite*{teller1973conditionalization, lewis1999papers, skyrms1987dynamic}). Crucially, this bookie has a significant amount of power: he offers bets both before and after the arrival of information, observes the information himself, and knows the agent's mistaken belief.

In this paper, in order to scrutinize and compare updating rules, we explore a similarly adversarial scenario.\footnote{In the same vein, \cite*{echenique2011money} and \cite*{lanier2024money} illustrate how money-pumps are useful in characterizing the magnitude of revealed preference violations.} Our environment contains a principal and an agent. They start with a common prior about an unknown state of the world, and the agent alone privately observes the realization of a signal. The agent updates her prior through a possibly non-Bayesian updating rule. Then, at this \textit{interim} stage, the principal offers the agent a decision problem consisting of a collection of actions and a utility function that delivers a state- and action-dependent payoff to the agent.

The principal's goal is to exploit the agent and make the agent's \textit{ex ante} expected payoff as low as possible. Moreover, because the principal is unaware of the agent's private information, he faces a \textit{screening} problem: generating a decision problem is just the offering to the agent of a menu of actions.

What makes the problem interesting is that the agent's choices are shaped by her possibly non-Bayesian posteriors. That is, upon observing a piece of information, she chooses an action according to her current belief, which is potentially produced by means other than Bayes' law. The principal knows the updating rule, but does not observe the signal realization.

We distinguish between two varieties of reaction to new information; overreaction and underreaction.

\medskip

\noindent \textbf{Overreaction:} Given a prior and a signal, an agent overreacts to a piece of information (a signal realization) if the Bayesian posterior lies on a line segment between the prior and the agent's posterior. That is, the agent overreacts to a signal realization if her posterior is further away from the prior than the Bayesian posterior. 

\medskip

\noindent \textbf{Underreaction:} An agent underreacts to a piece of information if her posterior lies on a line segment between the prior and the Bayesian posterior. Underreaction to a signal realization is when the agent's posterior lies closer to the prior than the Bayesian posterior.

\medskip

We discover that there is a massive difference between the agent's welfare in these two classes. If there exists a signal realization that produces an overly extreme posterior--formally, that produces a posterior that lies outside of the closed convex hull of the support of the Bayesian posteriors--the principal can exploit the agent to an arbitrarily large degree. Because the posterior is so extreme, the principal can offer the agent a single action that no Bayesian would take, but that the agent with the extreme (and incorrect) posterior would. By scaling up the penalties to taking this action in ``bad states'' the agent's expected loss is limitless. A corollary of this result is that if an agent overreacts to a signal realization, she can be exploited.

On the other hand, if an agent underreacts to information, if all of her posteriors are (weakly) closer to the prior than the Bayesian posteriors, she cannot be exploited. The agent's \textit{ex ante} expected payoff cannot be worse than that from her outside option (which we normalize to zero). The explanation for this impossibility is more subtle than that for the overreaction result. If the agent can be exploited, there must be some signal realization that leads her to take the ``wrong'' action from the perspective of the Bayesian agent who sees that signal realization. But for an underreacting agent, this means that this ``wrong action'' is attractive at the prior and therefore at least one other Bayesian posterior. To put differently, any action that hurts the agent following some signal realization benefits her at another, and the latter gain necessarily outweighs the former loss. This is a bit of an over-simplification but the essence is that the ``direction'' of a mistake for an underreacting agent precludes exploitation.

Our results suggest a comparative benefit to underreaction versus overreaction. Interestingly, \cite*{epstein2010non} also discover an advantage to underreaction, though in a completely different sense. They examine the long-run beliefs of an agent who repeatedly obtains information. An underreacting agent eventually arrives at the truth, whereas an overreacter may not. In contrast, we show that an underreacter's mistakes are much less costly than an overreacter's.

One of the justifications for studying updating rules other than Bayes' law is that people, by and large, are not perfect adherents to Bayes' law. Importantly, as stressed by \cite*{benjamin2019errors} and emphasized by \cite*{ba2022over}, in scenarios in which the environment is relatively transparent--namely, the prior and signal are explicitly specified--agents tend to underreact to information.\footnote{This stands in sharp contrast to the survey responses of professional forecasters, households, and investors, who overreact to news (\cite*{bordalo2022overreaction}). \cite*{ba2022over}'s thesis that this trend is driven by complexity is compelling.} As \cite*{benjamin2019errors} states bluntly, 

\medskip

\noindent ``\textbf{Stylized Fact 1} Underinference is by far the dominant direction of bias.''

\medskip

\noindent A contribution of our paper, therefore, is to provide a justification for why underreaction is so dominant--it is \textit{safer}.


\section{Screening By a Malevolent Principal}

There is an unknown state of the world \(\theta\), drawn according to a full-support prior \(\mu\) from some finite set of states \(\Theta\). The cardinality of \(\Theta\) is \(n \geq 2\) and \(\Delta\left(\Theta\right) \subset \mathbb{R}^{n-1}\) denotes the \(n\)-state probability simplex. An agent obtains information about the state by observing the realization of a signal \(\pi \colon \Theta \to \Delta (S)\), where \(S\) is a finite set of signal realizations. At the interim stage--\textit{viz.,} after the agent observes the signal realization but before the state is known--a principal offers the agent a menu of state dependent payoffs, i.e., a decision problem. 

Formally, the principal offers the agent \(\left(A,u\right)\), where \(A\) is a compact set of actions and \(u \colon A \times \Theta \to \mathbb{R}\) is a collection of state- and action-dependent payoffs (in utils). The agent selects at most one of the actions \(a \in A\), with her outside option's payoff--that that she obtains from taking none of the proffered actions--normalized to \(0\). We denote the outside option action by \(\varnothing\) and extend any \(u\) offered by the principal to \(A \cup \left\{\varnothing\right\}\) by defining \(u(\varnothing, \theta) = 0\) for all \(\theta \in \Theta\).

The principal shares the agent's prior, and given this, wishes to minimize the agent's \textit{ex ante} expected payoff. What makes the problem nontrivial, and what potentially gives the principal scope for exploiting the agent, is that the agent's updating rule may not be Bayes' law. That is, given any signal realization \(s\), the agent's posterior \(\hat{x}(s)\) may not equal \(x(s)\), where
\[x(s)(\theta) \coloneqq \frac{\mu(\theta) \pi\left(\left.s\right|\theta\right)}{\sum_{\theta \in \Theta} \mu(\theta) \pi\left(\left.s\right|\theta\right)}\]
is the Bayesian posterior. We specify that the signal \(\pi\) is known to the principal, as is how the agent interprets signal realizations,\footnote{That is, we specify that if the principal could observe \(s\) she would know the agent's posterior belief.} but that the signal realization, \(s\), itself is observed only by the agent.

Consequently, the principal faces a potentially non-trivial \textit{screening} problem. We specify further that \(\pi\) is such that the Bayesian posteriors are affinely independent,\footnote{A justification for this is that any Bayes-plausible (\cite*{kam}) distribution over (Bayesian) posteriors can be obtained by randomizing over Bayes-plausible distributions with affinely independent support.} and that each \(s \in S\) realizes with strictly positive probability. That is, \(\sum_{\theta \in \Theta} \mu(\theta) \pi\left(\left.s\right|\theta\right) > 0\) for all \(s \in S\).

All in all, here is the timing of the scenario:
\begin{enumerate}
    \item \(\theta\) is drawn according to \(\mu \in \inter \Delta(\Theta)\) and is observed by neither principal nor agent.
    \item \(s\) is drawn from \(S\) according to \(\pi\left(\left.\cdot\right|\theta\right)\) and is observed privately by the agent.
    \item The principal offers the agent a decision problem \(\left(A,u\right)\).
    \item The agent chooses a single element from \(A \cup \left\{\varnothing\right\}\).
    \item Payoffs realize.
\end{enumerate}
The agent's payoff is her \textit{ex ante} expected payoff
\[\sum_{\theta \in \Theta} \sum_{s \in S}\mu(\theta)\pi\left(\left.s\right|\theta\right) u(a^{*}(s),\theta)\text{,}\]
where \(a^*(s)\) is (without loss of generality)\footnote{Please refer to the discussion at the beginning of \(\S\)\ref{mainresults}.} the principal-preferred selection from
\[\argmax_{a \in A \cup \left\{\varnothing\right\}}\mathbb{E}_{\hat{x}(s)}u(a,\theta)\text{.}\]
The principal wishes to infimize the agent's payoff.

\begin{definition}
    Given \(\pi\) and \(\mu\), an agent can be \textcolor{ForestGreen}{Exploited} if for any \(K > 0\), the principal can offer the agent a decision problem yielding her an \textit{ex ante} expected payoff of \(-K\).
\end{definition}
Our imposition that the support of the distribution of Bayesian posteriors is affinely independent implies there are at most \(n\) posteriors. We let the agent's type denote her Bayesian posterior. Henceforth, in our notation, we also suppress the dependence of the posteriors on the signal realization, i.e., \(x_i = x(s_i)\) and \(\hat{x}_i = \hat{x}(s_i)\), where beliefs without hats are those produced by Bayes' law.

Let us make two easy observations. The first is an immediate consequence of the fact that the agent's payoff can be scaled without affecting incentives.
\begin{remark}
    An agent can be exploited if and only if the principal can offer the agent a decision problem yielding her an \textit{ex ante} expected payoff of \(-K\) for some \(K > 0\).
\end{remark}
The second remark is the revelation principle, more-or-less. \textit{Viz.}, we can always prune excess actions and reduce the agent's posterior-dependent choice to one that is deterministic.
\begin{remark}
    If an agent can be exploited, she can be exploited by a decision problem with at most \(n\) actions in which the agent's choice of action following any signal realization is deterministic.
\end{remark}

\section{Over- Versus Underreaction}

Given a collection of posteriors \(\left\{x_1, \dots, x_t\right\}\) let \(\conv \left\{x_1, \dots, x_t\right\}\) denote the closed convex hull of this set.
\begin{definition}
    Fix \(\mu\) and \(\pi\). An agent \textcolor{ForestGreen}{Underreacts to realization \(s\)} if there exists a \(\lambda \in \left[0,1\right]\) such that \(\hat{x}_s = \lambda \mu + (1-\lambda)x_{s}\). An agent \textcolor{ForestGreen}{Underreacts to Information} if she underreacts to each realization \(s \in S\).
\end{definition}
Our definition of overreaction to information is the easily-anticipated mirror:
\begin{definition}
    Fix \(\mu\) and \(\pi\). An agent \textcolor{ForestGreen}{Overreacts to realization \(s\)} if there exists a \(\lambda \in \left(0,1\right)\) such that \(x_{s} = \lambda \mu + (1-\lambda)\hat{x}_s\). An agent \textcolor{ForestGreen}{Overreacts to Information} if she overreacts to each realization \(s \in S\).
\end{definition}

\subsection{Some Non-Bayesian Updating Rules}

Before exploring the main results in \(\S\)\ref{mainresults}, let us briefly discuss some well-known updating rules. As highlighted by \cite*{declippel}, an oft-observed non-Bayesian rule observed in the literature is \textbf{\cite*{grether}'s \(\alpha-\beta\) Model of Updating} (see, e.g., \cite*{benjamin2019errors}, \cite*{benjamin2019base}, \cite*{augenblick2021belief}, and \cite*{raymond16}). As revealed by \cite*{declippel}, an \(\alpha-\beta\) distorted posterior \(\hat{x}_s\) is a function of the Bayesian posterior \(x_s\), which simplifies to
\[\hat{x}_s(x_s) = \frac{x_s^\beta}{x_s^\beta + (1-x_s)^\beta}\]
when there are just two states and the prior is \(1/2\).
Evidently, if \(\beta > 1\), the agent overreacts (\(x_s < \hat{x}_s\) if \(x_s > 1/2\) and \(x_s > \hat{x}_s\) if \(x_s < 1/2\)) and underreacts if \(\beta < 1\). 

An agent exhibits \textbf{Extreme Belief Aversion} (\cite*{ducharme1968intuitive}, \cite*{ducharme1970response}, and \cite*{raymond16}) if she is averse to holding beliefs close to certainty. This yields, from the perspective of our framework, underreaction to information, particularly when the agent's signal is extremely informative. Our main results, therefore, provide a rationale for such behavior:\footnote{\cite*{whitmeyer2023bayes} shows that when there are two states, an agent who is averse to extreme beliefs nevertheless prefers more information (\cite*{blackwell}, \cite*{blackwell2}) to less. Notably, extreme belief aversion is the only class of non-Bayesian updating rules compatible with the Blackwell order.} extreme belief aversion is relatively innocuous, and better than the opposite variety.

In \textbf{\cite*{rabin1999first}'s Model of Confirmatory Bias} (with two states and two signal realizations) an agent correctly processes the signal realization that confirms her prior but only obtains the correct posterior following the signal realization that contradicts her prior with some probability. With its complement, her posterior following this conflicting signal realization is the confirming-realization posterior instead. Formally, given \(\mu > 1/2\) and Bayesian posteriors \(x_H\) and \(x_L\) (\(0 \leq x_L < 1/2 < x_H \leq 1\)), \(\hat{x}_H = x_H\) and \(\hat{x}_L = x_H\) with some probability \(q \in \left(0,1\right)\) and \(\hat{x}_L = x_L\) otherwise. Though this type of random error does not fit within our main specification, we study a generalized version of this environment in \(\S\)\ref{rabin}.

An agent's malformed posterior could also come from having a \textbf{Misspecified Prior}, from the perspective of the principal, at least. Let us again specialize to two states and binary experiments. In this case, there is always one signal realization \(s\) with respect to which the agent overreacts. If she is not too misspecified, she underreacts to the other signal realization, but otherwise does not.


\subsection{Main Results}\label{mainresults}

Throughout this section, it is important to keep in mind that the agent chooses her action following a signal realization based on her posterior \(\hat{x}_s\) and not the Bayesian posterior \(x_s\). Furthermore, in principle, how we break ties when the agent is indifferent between actions could affect our results. Here is why how we break such ties does not matter. First, in our sufficiency results for when the agent can be exploited, in a decision problem that exploits the agent, incentives can be made strict for the posterior that matters, eliminated the importance of our tie-breaking rule. Second, our necessity results for the agent to be exploited, we break ties in a principal-optimal way as a different rule would only make it more difficult for him to exploit the agent.

Our first result identifies a sufficient condition for exploitation. Fixing \(\mu\) and \(\pi\), and letting \(X \coloneqq \left\{x_1, x_2, \dots, x_n\right\}\) denote the set of Bayesian posteriors,
\begin{lemma}\label{mainlemma}
    If there exists an \(s \in S\) for which \(\hat{x}_s \notin \conv X\), the agent can be exploited.
\end{lemma}
\begin{proof}
    Suppose there exists an \(s \in S\) for which \(\hat{x}_s \notin \conv X\). Without loss of generality \(\hat{x}_{s'} = x_{s'}\) for all \(s' \in S \setminus \left\{s\right\}\). As \(\hat{x}_s \notin \conv X\),  \(\hat{x}_s\) and \(\conv X\) can be strictly separated by a hyperplane
    \(H_{\alpha, \beta} \coloneqq \left\{x \in \left.\mathbb{R}^{n-1} \right| \alpha \cdot x = \beta\right\}\). Suppose the principal offers a single-action decision problem \(\left(\left\{a\right\}, u\right)\) where \(u\) is such that 
    \(\mathbb{E}_x u(a,\theta) = \lambda(\alpha \cdot x - \beta)\) (\(\lambda > 0\)) and \(\alpha \cdot \hat{x}_s - \beta > 0\). Observe that this implies \(\alpha \cdot x_{s'} - \beta < 0\) for all  \(s' \in S\). The agent's \textit{ex ante} expected payoff is 
    \[\lambda \sum_{\theta \in \Theta} \mu(\theta)\pi(\left.s\right|\theta) \left(\alpha \cdot x_{s} - \beta\right)\text{.}\]
    Then for any \(K > 0\) we just set 
    \[\lambda = \frac{-K}{\sum_{\theta \in \Theta} \mu(\theta)\pi(\left.s\right|\theta) \left(\alpha \cdot x_{s} - \beta\right)}\text{.}\]
    Exploitation achieved.\end{proof}
Indeed, observe that a single-action decision problem will do. The key part of the construction is that this action is strictly suboptimal at every one of the agent's Bayesian posteriors, yet her extreme posterior makes her strictly prefer to take it following some signal realization. \textbf{Big mistake}. We also wish to highlight that the affine independence of \(X\) is not needed for this result: so long as there is some \(\hat{x}_s\) outside of \(\conv X\), the result goes through.

This lemma implies a peril of overreaction. Note that this corollary \textit{does} rely on the affine independence of \(X\). If the set is not affinely independent, the agent can overreact to some \(s\) yet \(\hat{x}_s \in \conv X\).
\begin{corollary}
    Fix \(\mu\) and \(\pi\). If an agent overreacts to some realization \(s\), she can be exploited.
\end{corollary}
On the other hand, it is impossible to exploit an agent who underreacts to information. Letting \(\alpha^a \cdot x - \beta^a\) denote the expected payoff in belief \(x\) to action \(a\), here is the lemma central to establishing Proposition \ref{maintheorem}.
\begin{lemma}\label{prunelemma}
    If an underreacting agent can be exploited, there exists a contract that exploits the agent that does not contain any action \(a\) for which \(\alpha^a \cdot x_i - \beta^a > 0 \geq \alpha^a \cdot \hat{x}_i - \beta^a\) for some type \(x_i\).
\end{lemma}
\begin{proof}
    As we could just prune any actions that are not chosen by any type, we start with a decision problem that exploit the agent in which each action is chosen by at least one type. The set of actions available to the agent is \(A\). The following observation will be useful.
    \begin{claim}\label{claim1}
        If \(a \in A\) is such that \(\alpha^a \cdot x_i - \beta^a > 0 \geq \alpha^a \cdot \hat{x}_i - \beta^a\) for some type \(x_i\), we cannot have
        \(\alpha^a \cdot \hat{x}_j - \beta^a \geq 0 > \alpha^a \cdot x_j - \beta^a\) for any \(x_j \in X\).
    \end{claim}
    \begin{proof}
        Because the agent is underreacting,
        \(\alpha^a \cdot x_i - \beta^a > 0 \geq \alpha^a \cdot \hat{x}_i - \beta^a\) for some type \(x_i\) implies that
        \[0 \geq \alpha^a \cdot \hat{x}_i - \beta^a \geq \alpha^a \cdot \mu - \beta^a \text{.}\]
        Suppose for the sake of contradiction that \(\alpha^a \cdot \hat{x}_j - \beta^a \geq 0 > \alpha^a \cdot x_j - \beta^a\) for some \(x_j \in X\). Again, by underreaction, this implies 
        \[\alpha^a \cdot \mu - \beta^a \geq \alpha^a \cdot \hat{x}_j - \beta^a \geq 0 \text{,}\]
        a contradiction.\end{proof}
        Let \(A^{\dagger} \subseteq A\) denote the subset of actions in the agent's decision problem for which \(\alpha^a \cdot x_i - \beta^a > 0 \geq \alpha^a \cdot \hat{x}_i - \beta^a\) for some type \(x_i\), i.e., formally,
        \[A^{\dagger} \coloneqq \left\{a \in \left.A\right\vert \exists x_i \in X \colon \alpha^a \cdot x_i - \beta^a > 0 \geq \alpha^a \cdot \hat{x}_i - \beta^a\right\}\text{.}\]
        Without loss of generality we assume that \(A^{\dagger}\) is non-empty, or else we are done.

        Claim \ref{claim1} implies that for any \(a^{\dagger} \in A^{\dagger}\) and each \(x_i \in X\), either
        \[\label{claimlist}\tag{{\color{OrangeRed} \(\clubsuit\)}}
        \begin{split}
            &1. \quad \alpha^{a^{\dagger}} \cdot x_i - \beta^{a^{\dagger}} > 0 \geq \alpha^{a^{\dagger}} \cdot \hat{x}_i - \beta^{a^{\dagger}} \ \text{; or}\\
            &2. \quad 0 > \alpha^{a^{\dagger}} \cdot x_i - \beta^{a^{\dagger}} \geq \alpha^{a^{\dagger}} \cdot \hat{x}_i - \beta^{a^{\dagger}} \ \text{; or}\\
            &3. \quad \alpha^{a^{\dagger}} \cdot x_i - \beta^{a^{\dagger}} \geq \alpha^{a^{\dagger}} \cdot \hat{x}_i - \beta^{a^{\dagger}} > 0 \text{.}
        \end{split}\]
        Immediately, we see that we cannot have \(A^{\dagger} = A\), or else the agent's \textit{ex ante} expected payoff would be weakly positive, contradicting that she is being exploited. Accordingly \(A^{\dagger} \subset A\) (but remember \(A^{\dagger}\) is non-empty). Moreover for any \(a \in A \setminus A^{\dagger}\) and any \(x_i \in X\), either
        \[\label{claimlist2}\tag{\textcolor{OrangeRed}{\(\diamondsuit\)}}
        \begin{split}
            &1. \quad \alpha^a \cdot \hat{x}_i - \beta^a \geq 0 > \alpha^a \cdot x_i - \beta^a \ \text{; or}\\
            &2. \quad 0 > \alpha^{a} \cdot \hat{x}_i - \beta^{a} \geq \alpha^{a} \cdot x_i - \beta^{a} \ \text{; or}\\
            &3. \quad \alpha^{a} \cdot \hat{x}_i - \beta^{a} \geq \alpha^{a} \cdot x_i - \beta^{a} > 0 \text{.}
        \end{split}\]

        Because the agent is being exploited, we must have \[\label{eqpreprune}\tag{{\color{OrangeRed} \(\heartsuit\)}} 0 > \sum_{i=1}^{k} p_i \left(\alpha_i \cdot x_i - \beta_i\right) + {\color{SeaGreen} \sum_{i=k+1}^{n} p_i \left(\alpha_i \cdot x_i - \beta_i\right)}\text{,}\]
        where \(\alpha_i \cdot x - \beta_i\) is the expected payoff to type \(x_i\) from choosing optimally at belief \(\hat{x}_i\). Keep in mind that it is possible that a type \(x_i\) is choosing the outside option, in which case \(\alpha_i = \underbar{0}\) and \(\beta_i = 0\). Moreover, the set of types \(\left\{x_1, \dots, x_k\right\}\) (\(1 \leq k < n\)) are those whose choice of action (assuming any indifferences are broken in a principal-optimal way) is different when the set of available actions is \(A \setminus A^{\dagger}\) rather than \(A\).

        Suppose that the principal prunes away the set of actions \(A^{\dagger}\), offering the agent \(A \setminus A^{\dagger}\) instead. Observe that agent's payoff (breaking ties in a principal-optimal way) is now
        \[\label{newpay}\tag{\textcolor{OrangeRed}{\(\spadesuit\)}} \sum_{i=1}^{k} p_i \left(\tilde{\alpha}_i \cdot x_i - \tilde{\beta}_i\right) + {\color{SeaGreen} \sum_{i=k+1}^{n} p_i \left(\alpha_i \cdot x_i - \beta_i\right)}\text{,}\]
        where \(\tilde{\alpha}_i \cdot x - \tilde{\beta}_i\) is the expected payoff to type \(x_i\) (\(i \leq k\)) from choosing optimally at belief \(\hat{x}_i\) when the set of available actions is \(A \setminus A^{\dagger}\). 
        
        The green sums in Expressions \ref{eqpreprune} and \ref{newpay} are the same. Moreover, by the Consequences of Claim \ref{claim1} enumerated in Lists \ref{claimlist} and \ref{claimlist2}, we must have
        \[\tilde{\alpha}_i \cdot x_i - \tilde{\beta}_i \leq \tilde{\alpha}_i \cdot \hat{x}_i - \tilde{\beta}_i \leq \alpha_i \cdot \hat{x}_i - \beta_i \leq \alpha_i \cdot x_i - \beta_i\text{,}\]
        for each \(i \leq k\) and so
         \[\sum_{i=1}^{k} p_i \left(\tilde{\alpha}_i \cdot x_i - \tilde{\beta}_i\right) + {\color{SeaGreen} \sum_{i=k+1}^{n} p_i \left(\alpha_i \cdot x_i - \beta_i\right)} \leq \sum_{i=1}^{k} p_i \left(\alpha_i \cdot x_i - \beta_i\right) + {\color{SeaGreen} \sum_{i=k+1}^{n} p_i \left(\alpha_i \cdot x_i - \beta_i\right)} < 0 \text{,}\]
         completing the proof.\end{proof}

\begin{theorem}\label{maintheorem}
    Fix \(\mu\) and \(\pi\). If an agent underreacts to information, she cannot be exploited.
\end{theorem}
\begin{proof}
    Suppose for all \(s \in S\) \(\hat{x}_s\) lies on the line segment between \(\mu\) and \(x_{s}\). 
    
    Our proof approach is via induction. We shall first show that a single-action decision problem is incapable of exploiting the underreacting agent. We will then argue that if a \(t\)-action (\(1 \leq t\)) decision problem cannot exploit the agent, neither can a \((t+1)\)-action decision problem.

    \medskip

    \noindent \textbf{Step 1. Base case: A single-action decision problem cannot exploit the agent.} As each \(\hat{x}_s\) lies on the line segment between \(\mu\) and \(x_{s}\), for any hyperplane \(H_{\alpha, \beta} \coloneqq \left\{x \in \left.\mathbb{R}^{n-1} \right| \alpha \cdot x = \beta\right\}\) that strictly separates some \(x_{s'}\) from \(\hat{x}_{s'}\), for which \(\alpha \cdot \hat{x}_{s'} - \beta > 0\), we must have \(\alpha \cdot \mu - \beta > 0\). Accordingly, as the average posterior is the prior \(\mu\), the agent's payoff is bounded (weakly) below by \(\alpha \cdot \mu - \beta > 0\).

    \medskip
    
    \noindent \textbf{Step 2. Induction hypothesis and step: if a \(t\)-action (\(1 \leq t\)) decision problem cannot exploit the agent, neither can a \((t+1)\)-action decision problem.} Suppose for the sake of contradiction that there exists a \((t+1)\)-action decision problem that exploits the agent. By Lemma \ref{prunelemma}, we assume without loss of generality that there does not exist an action \(a\) for which \(\alpha^a \cdot x_i - \beta^a > 0 \geq \alpha^a \cdot \hat{x}_i - \beta^a\) for some type \(x_i\). 
    
    Denote the actions in this decision problem by \(\tilde{a}, a^1, \dots, a^t\), keeping in mind that the set of actions available to the agent is \(\left\{\tilde{a}, a^1, \dots, a^t, \varnothing\right\}\). In the specified \((t+1)\)-action decision problem, let \(\left\{x_1, \dots, x_k\right\}\) (\(1 \leq k\)) be the types who select action \(\tilde{a}\) and \(\left\{x_s, \dots, x_n\right\}\), those who select \(\varnothing\). By assumption,
    \[\label{eq1}\tag{\textcolor{OrangeRed}{\(\bigstar\)}}0 > \sum_{i = 1}^{k} p_i (\tilde{\alpha} \cdot x_i - \tilde{\beta}) + {\color{SeaGreen} \sum_{i=k+1}^{s-1} p_i (\alpha_i \cdot x_i - \beta_i)}\text{,}\]
    where \(\tilde{\alpha} \cdot x - \tilde{\beta}\) is the expected payoff from taking action \(\tilde{a}\) and \(\alpha_i \cdot x - \beta_i\) is the expected payoff from taking the optimal action for type \(x_i\) given her posterior \(\hat{x}_i\).

    As no \(t\) action decision problem exploits the agent, we must have (when \(\tilde{a}\) is removed) \[\label{eq2}\tag{\textcolor{OrangeRed}{\(\blacklozenge\)}}\sum_{i = 1}^{k} p_i (\alpha_i \cdot x_i - \beta_i) + {\color{SeaGreen}\sum_{i=k+1}^{s-1} p_i (\alpha_i \cdot x_i - \beta_i)} \geq 0\text{.}\]
    where \(\alpha_i \cdot x - \beta_i\) is the expected payoff, in the posterior \(x\), for type \(x_i\)'s optimal action in the set \(\left\{a^1, \dots, a^t, \varnothing\right\}\) given her belief \(\hat{x}_i\). Note that green expressions in Inequalities \ref{eq1} and \ref{eq2} are the same.
    
    Analogously, as our choice of an outside option of \(0\) was merely a normalization,
    \[\sum_{i = 1}^{k} p_i (\tilde{\alpha} x_i - \tilde{\beta}) + \sum_{i=k+1}^{s-1} p_i (\alpha_i \cdot x_i - \beta_i) + \sum_{i=s}^n p_i (\alpha_i \cdot x_i - \beta_i) \geq \tilde{\alpha} \cdot \mu - \tilde{\beta} \geq 0\text{.}\]
    
    Crucially, for each \(i \geq s\), \(\alpha_i \cdot \hat{x}_i - \beta_i \leq 0\). Thus, for each \(i \geq s\), \(\alpha_i \cdot x_i - \beta_i \leq 0\), which implies 
    \[\sum_{i = 1}^{k} p_i (\tilde{\alpha} x_i - \tilde{\beta}) + \sum_{i=k+1}^{s-1} p_i (\alpha_i \cdot x_i - \beta_i) \geq 0\text{,}\] a contradiction.\end{proof}
    It turns out that if we specialize to experiments with just two signal realizations, we can completely characterize updating behavior that leaves an agent vulnerable (or not) to exploitation.
    \begin{proposition}\label{binaryprop}
        Fix \(\mu\) and binary \(\pi\). An agent cannot be exploited if and only if she underreacts to information.
    \end{proposition}
    \begin{proof}
        Sufficiency follows from Theorem \ref{maintheorem}. By Lemma \ref{mainlemma}, we need only check the situation in which one of the two posteriors \(\hat{x}_0\) can be strictly separated from \(\conv\left\{x_0, \mu\right\}\) by some hyperplane \(\alpha \cdot x = \beta\). However, in this instance, there exists a single-action decision problem, such that the agent's payoff is bounded above by \(\alpha \cdot \mu - \beta < 0\). \end{proof}
        This proposition levers the fact that if the agent's experiment is binary, and if an agent's posterior lies within the convex hull of the Bayesian posteriors, it is either an under-reaction, or it ``skips over the prior.'' As the proof reveals, such a skip renders the agent vulnerable. In the next section, we use this proposition to shed light on a broad class of updating rules.

\subsection{Systematic Distortions}

A special class of updating rules are those for which the posterior produced by the updating rule is a function of the Bayesian posterior, irrespective of the prior or the signal. 
\begin{definition}
    An updating rule systematically distorts beliefs if for any full-support prior \(\mu\), signal \(\pi\), and signal realization \(s\), the posterior produced by the agent's updating rule \(\hat{x}_s\) equals \(\varphi(x_s)\) for some function \(\varphi \colon \Delta(\Theta) \to \Delta(\Theta)\).
\end{definition}
\cite*{declippel} introduce this class of updating rules, which has subsequently been studied by, e.g., \cite*{chambers2023coherent} and \cite*{whitmeyer2023bayes}. \cite*{coarsebayes} explores a particular class of such systematic distortions.

\begin{definition}
    An agent is \textcolor{ForestGreen}{Exploitable} if for any \(K > 0\), there exists a full-support prior \(\mu\), a signal \(\pi\), and a decision problem yielding the agent an \textit{ex ante} expected payoff of \(-K\).
\end{definition}
We term an agent who is not exploitable \textcolor{ForestGreen}{\textit{Unexploitable}}. Repeating our earlier remark almost verbatim, and via the exact same logic,
\begin{remark}
    An agent is exploitable if and only if the principal can offer the agent a decision problem yielding her an \textit{ex ante} expected payoff of \(-K\) for some \(K > 0\).
\end{remark}
Then,
\begin{proposition}
    An agent whose updating rule systematically distorts beliefs is unexploitable if and only if for any full-support prior \(\mu\) and signal \(\pi\) she underreacts to information.
\end{proposition}
\begin{proof}
    \(\left(\Leftarrow\right)\) Suppose for the sake of contraposition there exists a full-support prior \(\mu\) and signal \(\pi\) for which the agent does not underreact to information. Because the agent's updating rule systematically distorts beliefs, this implies there exists a binary experiment and a signal realization \(s\) for which \(\hat{x}_s\) does not lie within the closed convex hull of the set of Bayesian posteriors. Accordingly, Lemma \ref{mainlemma} implies that the agent can be exploited and so she is exploitable.

    \medskip

    \noindent \(\left(\Rightarrow\right)\) This is an immediate consequence of Theorem \ref{maintheorem}.
\end{proof}

\subsection{Uncertain Rules and Random Mistakes}

It is possible that the principal does not know the updating rule used by the agent, or that agent's updating rule maps at least one signal realization at random to posteriors. These two cases are formally equivalent: for each signal realization \(s \in S\), the agent's posterior is a random vector \(\hat{\mathcal{X}}_s\) distributed according to some \(\gamma_s\) on \(\Delta(\Theta)\). Recalling that \(X\) denotes the set of Bayesian posteriors and \(\conv X\) its closed convex hull, we have
\begin{lemma}
    Fix \(\mu\) and \(\pi\). If there exists a closed convex set \(B \subset \Delta(\Theta)\) disjoint from \(X\) for which \(\mathbb{P}(\hat{\mathcal{X}}_s) \in B > 0\), the agent can be exploited.
\end{lemma}
\begin{proof}
By the strict separating hyperplane theorem, there exists a hyperplane that strictly separates \(B\) from \(X\). Using this, the remainder of the proof mirrors that for Lemma \ref{mainlemma} and so we omit it.
\end{proof}
It is also easy to derive an analog of Theorem \ref{maintheorem}.
\begin{corollary}
    Fix \(\mu\) and \(\pi\). If for each \(s\) \(\hat{\mathcal{X}}_s\) is supported on a finite subset of \(\conv\left\{x_s, \mu\right\}\), the agent cannot be exploited.
\end{corollary}
\begin{proof}
    This result is exactly Theorem \ref{maintheorem} once we redefine an agent's type to be her realized posterior. Given this, there will obviously be more than \(n\) types (unless the updating is deterministic), but there are still only finitely many types. \end{proof}

\subsubsection{Rabin and Schrag's Model of Confirmatory Bias}\label{rabin}

One type of updating rule that does not fit our main specification, nor the extension of this subsection is \cite*{rabin1999first}'s model of confirmatory bias. A more general version of this is as follows: there are \(n\) states, a full-support prior \(\mu\) and a signal \(\pi\) that produces a collection of affinely-independent Bayesian posteriors. Then, for any signal realization \(s\), the agent's posterior is a random variable \(\hat{\mathcal{X}}_s\) supported on the Bayesian posteriors \(X = \left\{x_1, \dots, x_n\right\}\). For each \(s\), we let \(q_s \coloneqq \mathbb{P}(\hat{\mathcal{X}}_s = x_s) \in \left[0,1\right]\), i.e., \(q_s\) is the probability \(s\) is interpreted correctly. We call this class of updating rules \textcolor{ForestGreen}{Generalized Confirmatory Bias}.

Any error in this sense allows the principal to exploit the agent.
\begin{proposition}
    Fix \(\mu\) and \(\pi\). If the agent's updating rule is generalized confirmatory bias and there exists an \(s\) for which \(q_s < 1\), the agent can be exploited.
\end{proposition}
\begin{proof}
    Suppose there exists an \(s\) for which \(q_s < 1\). Then there exists an \(s' \neq s\) for which \(\mathbb{P}(\hat{\mathcal{X}}_s = x_{s'}) > 0\). As \(x_{s'}\) is an exposed point of \(X\), there exists a supporting hyperplane \(H_{\alpha, \beta} \coloneqq \left\{x \in \left.\mathbb{R}^{n-1} \right| \alpha \cdot x = \beta\right\}\) for which \(H_{\alpha, \beta} \cap \conv X = \left\{x_{s'}\right\}\). Accordingly, there exists an action yielding payoff \(\alpha \cdot x - \beta\), where \(\alpha \cdot x_s - \beta < 0\) for all \(s \neq s'\) and \(\alpha \cdot x_{s'} - \beta = 0\). As the agent's payoff is the average of a strictly negative number with \(0\), her expected payoff is strictly negative; that is, she can be exploited.\end{proof}

\bibliography{sample.bib}

\end{document}